\documentclass{acm_proc_article-sp}

\usepackage{amsthm}
\usepackage{amsmath, amssymb}
\usepackage{thm-restate}
\usepackage{graphicx}
\usepackage{subfigure, natbib}

\usepackage{tabularx, rotating, slashbox}

\usepackage{paralist}

\usepackage[colorlinks=true,citecolor=blue,linkcolor=blue,urlcolor=blue]{hyperref} 
\usepackage{nameref}

\makeatletter
\let\orgdescriptionlabel\descriptionlabel
\renewcommand*{\descriptionlabel}[1]{%
  \let\orglabel\label
  \let\label\@gobble
  \phantomsection
  \edef\@currentlabel{#1}%
  \let\label\orglabel
  \orgdescriptionlabel{#1}%
}
\makeatother

\usepackage{tikz}
\usetikzlibrary{arrows,%
                shapes,positioning}
\usepackage{myDefs}  

\newtheorem{theorem}{Theorem}
\newtheorem{lemma}[theorem]{Lemma}

\newtheorem{notation}[theorem]{Notation}
\newtheorem{definition}[theorem]{Definition}

\newcommand{\axioms}{Axioms (1-8)}
\newcommand{\ts}{tie strength }

\begin{document}

\title{Measuring Tie Strength in Implicit Social Networks}
  
\numberofauthors{2} 

\author{
  \alignauthor
  Mangesh Gupte\titlenote{Current affiliation: Google.}\\
  \affaddr{Department of Computer Science}\\
  \affaddr{Rutgers University}\\
  \affaddr{Piscataway, NJ 08854}\\
  \email{mangesh@cs.rutgers.edu}
  \alignauthor
  Tina Eliassi-Rad\\
  \affaddr{Department of Computer Science}\\
  \affaddr{Rutgers University}\\
  \affaddr{Piscataway, NJ 08854}\\
  \email{tina@eliassi.org}
}

\date{\today}

\maketitle

\begin{abstract}
  Given a set of people and a set of events they attend, we address the problem
  of measuring \emph{connectedness} or \emph{tie strength} between each pair of
  persons given that attendance at mutual events gives an implicit social
  network between people.
  We take an axiomatic approach to this problem. Starting from a list of axioms
  that a measure of tie strength must satisfy, we characterize functions that
  satisfy all the axioms and show that there is a range of measures that satisfy
  this characterization. A measure of tie strength induces a ranking on the
  edges (and on the set of neighbors for every person). We show that for
  applications where the ranking, and not the absolute value of the tie
  strength, is the important thing about the measure, the axioms are equivalent
  to a natural partial order. Also, to settle on a particular measure, we must
  make a non-obvious decision about extending this partial order to a total
  order, and that this decision is best left to particular applications. We
  classify measures found in prior literature according to the axioms that they
  satisfy.
  In our experiments, we measure tie strength and the coverage of our axioms in several datasets. Also, for each
  dataset, we bound the maximum Kendall's Tau divergence (which measures the
  number of pairwise disagreements between two lists) between all measures that
  satisfy the axioms using the partial order. This informs us if particular
  datasets are well behaved where we do not have to worry about which measure to
  choose, or we have to be careful about the exact choice of measure we make.
\end{abstract}


\section{Introduction}
\label{sec:intro}

Explicitly declared friendship links suffer from a low signal to noise ratio
(e.g. Facebook friends or LinkedIn contacts). Links are added for a variety of
reasons like reciprocation, peer-pressure, etc. Detecting which of these links
are important is a challenge.

Social structures are implied by various interactions between users of a
network. We look at event information, where users participate in mutual
events. Our goal is to infer the strength of ties between various users given
this event information. Hence, these social networks are implicit.

There has been a surge of interest in implicit social networks. We can see
anecdotal evidence for this in startups like COLOR (\url{http://www.color.com})
and new features in products like Gmail. COLOR builds an implicit social network
based on people's proximity information while taking
photos.\footnote{\url{http://mashable.com/2011/03/24/color/}} Gmail's
\emph{don't forget bob}~\cite{google_dont_forget_bob} feature uses an implicit
social network to suggest new people to add to an email given a existing list.

People attend different events with each other. In fact, an event is defined by
the set of people that attend it. An event can represent the set of people who
took a photo at the same place and time, like COLOR, or a set of people who are
on an email, like in Gmail. Given the set of events, we would like to infer how
\emph{connected} two people are, i.e. we would like to measure the
\emph{strength of the tie} between people. All that is known about each event is
the list of people who attended it. People attend events based on an implicit
social network with ties between pairs of people. We want to solve the inference
problem of finding this weighted social network that gives rise to the set of
events.

Given a bipartite graph, with people as one set of vertices and events as the
other set, we want to infer the tie-strength between the set of people. Hence,
in our problem, we do not even have access to any directly declared social
network between people, in fact, the social network is implicit. We want to
infer the network based on the set of people who interact together at different
points in time.

We start with a set of axioms and find a characterization of functions that
could serve as a measure of tie strength, just given the event information. We
do not end up with a single function that works best under all circumstances,
and in fact we show that there are non-obvious decisions that need to be made to
settle down on a single measure of tie strength.

Moreover, we examine  the case where the absolute
value of the tie strength is not important, just the order is important (see Section~\ref{sec:partial_order}). We show
that in this case the axioms are equivalent to a natural partial order on the
strength of ties. We also show that choosing a particular tie strength function
is equivalent to choosing a particular linear extension of this partial order.

Our contributions are:
\vspace*{-12pt}
\begin{itemize}
\item \emph{We present an axiomatic approach to the problem of inferring
    implicit social networks by measuring tie strength.}
\item \emph{We characterize functions that satisfy all the axioms and show a
    range of measures that satisfy this characterization.}
\item \emph{We show that in ranking applications, the axioms are equivalent to a
    natural partial order; we demonstrate that to settle on a particular
    measure, we must make non-obvious decisions about extending this partial
    order to a total order which is best left to the particular application.}
\item \emph{We classify measures found in prior literature according to the
    axioms that they satisfy.}
\item \emph{In our experiments, we show that by using Kendall's Tau divergence,
    we can judge whether a dataset is well-behaved, where we do not have to
    worry about which tie-strength measure to choose, or we have to be careful
    about the exact choice of measure.}
\end{itemize}
\vspace*{-6pt}

The remainder of this paper is structured as follows.
Section~\ref{sec:related} outlines the related work.  
Section~\ref{sec:model} presents our proposed model.
Sections~\ref{sec:axioms} and~\ref{sec:measures} describe 
the axioms and measures of tie strength, respectively.  
Section~\ref{sec:experiments} presents our experiments.  
Section~\ref{sec:conclusions} concludes the paper.

\section{Related Work}
\label{sec:related}
\citep*{gran} introduced the notion of strength of ties in social networks and
since then has affected different areas of study. We split the related works
into different subsections that emphasize particular methods/applications.

{\bf Strength of Ties: } \citep*{gran} showed that weak ties are important for
various aspects like spread of information in social networks. There have been
various studies on identifying the strength of ties given different features of
a graph. \citep*{predictingtiestrength} model tie strength as a linear
combination of node attributes like intensity, intimacy, etc to classify ties in
a social network as strong or weak. The weights on each attribute enable them to
find attributes that are most useful in making these
predictions. \citep*{transactional_tie_strength} take a supervised learning
approach to the problem by constructing a predictor that determines whether a
link in a social network is a strong tie or a weak tie. They report that
\emph{network transactional features}, which combine network structure with
transactional features like the number of wall posting, photos, etc like
$\frac{|posts(i,j)|}{\Sigma_{k}|posts_(j,k)|}$, form the best predictors.

{\bf Link Prediction: } \citep*{adamic_score} considers the problem of
predicting links between web-pages of individuals, using information such as
membership of mailing lists and use of common phrases on web pages. They define
a measure of similarity between users by creating a bipartite graph of users on
the left and features (e.g., phrases and mailing-lists) on the right as $w(u,v)
= \sum_{(i \text{ neighbor of } u \& v)}\reci{\log|i|}$.  \citep*{link_prediction}
formalizes the problem of predicting which new interactions will occur in a
social network given a snapshot of the current state of the network. It uses
many existing predictors of similarity between nodes like \citep*{adamic_score,
  simrank, katz} and generates a ranking of pairs of nodes that are currently not
connected by an edge. It compares across different datasets to measure the
efficacy of these measures. Its main finding is that there is enough information
in the network structure that all the predictors handily beat the random
predictor, but not enough that the absolute number of predictions is high.
\citep*{bipartite_internal} addresses the problem of predicting links in a
bipartite network. They define \emph{internal links} as links between left nodes
that have a right node in common, i.e. they are at a distance two from each
other and the predictions that are offered are only for internal links.

{\bf Email networks: } Because of the ubiquitous nature of email, there has been
a lot of work on various aspects of email
networks. \citep*{google_dont_forget_bob} discusses a way to suggest more
recipients for an email given the sender and the current set of recipients. This
feature has been integrated in the Google's popular Gmail
service. \citep*{transactional_tie_strength} constructs a regression model for
classifying edges in a social network as strong or weak. They achieve high
accuracy and find that \emph{network-transactional} features like number of
posts from $u$ to $v$ normalized by the total number of posts by $u$ achieve the
largest gain in accuracy of prediction.

{\bf Axiomatic approach to Similarity:} 
\citep*{pagerank_axioms} were one of the first to axiomatize graph measures.  In particular, they studied axiomatizing PageRank.  The closest in spirit to our work is the
work by Lin~\citep*{lin_info-theory-def-sim} that defines an information
theoretic measure of similarity. This measure depends on the existence of a
probability distribution on the features that define objects. While the measure
of tie strength between people is similar to a measure of similarity, there are
important differences. We do not have any probability distribution over events,
just a log of the ones that occurred. More importantly,
\citep*{lin_info-theory-def-sim} defines items by the attributes or features
they have. Hence, items with the same features are identical. In our case, even
if two people attend all the same events, they are not the same person, and in
fact they might not even have very high tie strength depending on how large the
events were. 

\section{Model}
\label{sec:model}

We model people and events as nodes and use a bipartite graph $G = (L \cup R, E)$
where the edges represent membership. The left vertices correspond to people
while the right vertices correspond to events. We ignore any information other
than the set of people who attended the events, like the timing, location,
importance of events. These are features that would be important to the overall
goal of measuring tie strength between users, but in this work we focus on the
task of inferring \ts using the graph structure only. We shall denote users in
$L$ by small letters ($u,v,\ldots$) and events in $R$ by capital
letters($P,Q,\ldots$). There is an edge between $u$ and $P$ if and only if $u$
attended event $P$. Hence, our problem is to find a function on bipartite graphs
that models \emph{tie strength} between people, given this bipartite graph
representation of events.

\begin{figure}[tbh]
  \centering
  \includegraphics[scale=0.42]{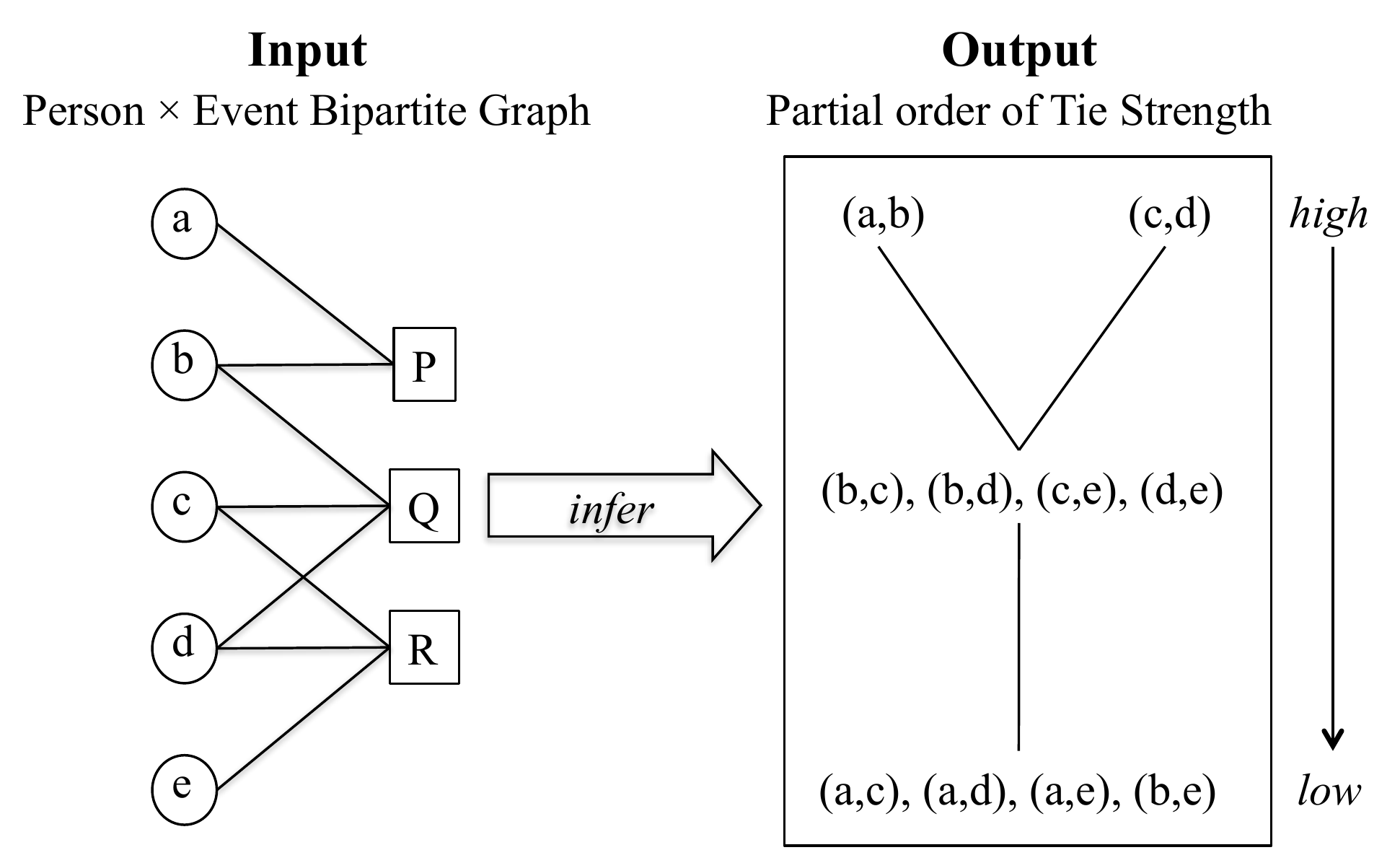}
  \caption{Given a bipartite person $\times$ event graph, we want to infer the induced partial order
    of tie strength among the people.}
  \label{fig:example}
\end{figure}

We also introduce some notation. We shall denote the tie strength of $u$ and $v$
due to a graph $G$ as $TS_G(u,v)$ or as $TS(u,v)$ if $G$ is obvious from
context. We shall also use $TS_{\{E_1,\ldots,E_k\}}(u,v)$ to denote the tie
strength between $u$ and $v$ in the graph induced by events
${\{E_1,\ldots,E_k\}}$ and users that attend at least one of these events. For a
single event $E$, then $TS_E(u,v)$ denotes the \ts between $u$ and $v$ if $E$
where the only event.

We denote the set of natural numbers by $\natnum$. A sequence of $k$ natural
numbers is given by $(a_1, \ldots, a_k)$ and the set of all such sequences is
$\natnum^k$. The set of all finite sequence of natural numbers is represented as
$\natnum^* = \cup_k \natnum^k$

\section{Axioms of Tie Strength}
\label{sec:axioms}
We now discuss the axioms that measures of tie strength between two users $u$
and $v$ must follow.

\begin{description}
\item[Axiom 1 (Isomorphism)\label{axiom:iso}] Suppose we have two graphs $G$ and
  $H$ and a mapping of vertices such that $G$ and $H$ are isomorphic. Let vertex
  $u$ of $G$ map to vertex $a$ of $H$ and vertex $v$ to $b$. Then $TS_G(u,v) =
  TS_H(a,b)$. Hence, the tie strength between $u$ and $v$ does not depend on the
  labels of $u$ and $v$, only on the link structure.  
\item[Axiom 2 (Baseline)\label{axiom:baseline}] If there are no events, then the
  tie strength between each pair $u$ and $v$ is 0. $TS_\phi(u,v) = 0$. If there
  are only two people $u$ and $v$ and a single party which they attend, then
  their tie strength is $1$. $TS_{\{u,v\}}(u,v) = 1$.
\item[Axiom 3 (Frequency: More events create stronger ties)\label{axiom:more_stronger}] All
  other things being equal, the more events common to $u$ and $v$, the stronger
  the tie strength of $u$ and $v$. \\Given a graph $G = (L\cup R, E)$ and two
  vertices $u,v \in L$. Consider the graph $G' = (L\cup (R\cup P), E \cup
  P_{u,v,\ldots})$, where $P_{u,v,\ldots}$ is a new event which both $u$ and $v$
  attend. Then the $TS_{G'}(u,v) \geq TS_G(u,v)$.
\item[Axiom 4 (Intimacy: Smaller events create stronger
  ties) \label{axiom:smaller_stronger}] All other things being equal, the fewer
  invitees there are to any particular party attended by $u$ and $v$, the
  stronger the \ts between $u$ and $v$. \\ Given a graph $G = (L\cup R, E)$ such
  that $P \in R$ and $(P,u),(P,v), (P,w) \in E$ for some vertex $w$. Consider
  the graph $G' = (L\cup R), E-(P,w))$, where the edge $(P,w)$ is deleted. Then
  the $TS_{G}(u,v) \geq TS_{G'}(u,v)$.
\item[Axiom 5 (Larger events create more ties)\label{axiom:larger_more}]
  Consider two events $P$ and $Q$. If the number of people attending $P$ is
  larger than the number of people attending $Q$, then the total \ts created by
  event $P$ is more than that created by event $Q$. \\$|P| \geq |Q| \implies
  \sum_{u,v \in P}TS_P(u,v) \geq \sum_{u,v \in Q} TS_Q(u,v)$.
\item[Axiom 6 (Conditional Independence of Vertices)\label{axiom:depend}] The tie
  strength of a vertex $u$ to other vertices does not depend on events that $u$
  does not attend; it only depends on events that $u$ attends.
\item[Axiom 7 (Conditional Independence of Events)\label{axiom:indep}] The
  increase in tie strength between $u$ and $v$ due to an event $P$ does not
  depend other events, just on the existing tie strength between $u$ and
  $v$. \\$TS_{G+P}(u,v) = g(TS_G(u,v), TS_P(u,v))$ for some fixed function
  monotonically increasing function $g$.
\item[Axiom 8 (Submodularity)\label{axiom:submodularity}] The marginal increase
  in tie strength of $u$ and $v$ due to an event $Q$ is at most the tie strength between $u$ and $v$ if $Q$ was their only event. 
  \\If $G$ is a graph and $Q$ is a single
  event, $TS_G(u,v) + TS_Q(u,v) \geq TS_{G+Q}(u,v)$.
\end{description}

\subsection*{Discussion}
These axioms give a measure of tie strength between nodes that is positive but
unbounded. Nodes that have a higher value are closer to each other than nodes
that have lower value. 

We get a sense of the axioms by applying them to
Figure~\ref{fig:example}. \ref{axiom:iso} implies that $TS(b,c) = TS(b,d)$ and
$TS(c,e) = TS(d,e)$. \ref{axiom:baseline}, \ref{axiom:depend} and
\ref{axiom:indep} imply that $TS(a,c) = TS(a,d) = TS(a,e) = TS(b,e) =
0$. \ref{axiom:smaller_stronger} implies that $TS(a,b) \geq TS(d,e)$.
\ref{axiom:more_stronger} implies that $TS(c,d) \geq TS(d,e)$.

While each of the axioms above are fairly intuitive, they are hardly trivial. In
fact, we shall see that various measures used in prior literature break some of
these axioms. On the other hand, it might seem that satisfying all the axioms is
a fairly strict condition. However, we shall see that even satisfying all the
axioms are not sufficient to uniquely identify a measure of tie strength. The axioms
leave considerable space for different measures of tie strength.

One reason the axioms do not define a particular function is that there is
inherent tension between \ref{axiom:smaller_stronger} and
\ref{axiom:more_stronger}. While both state ways in which \ts becomes
stronger, the axioms do not resolve which one dominates the other or how they
interact with each other. This is a non-obvious decision that we feel is best
left to the application in question. In Figure~\ref{fig:example}, we cannot tell
using just \axioms\; which of $TS(a,b)$ and $TS(c,d)$ is larger. We discuss this
more more in Section~\ref{sec:order}.

\subsection{Characterizing Tie Strength}
\label{sec:char}
In this section, we shall state and prove Theorem~\ref{thm:char} that gives a
characterization of all functions that satisfy the axioms of tie
strength. \axioms\; do not uniquely define a function, and in fact, one of the
reasons that tie strength is not uniquely defined up to the given axioms is that
we do not have any notion for comparing the relative importance of number of
events (frequency) versus the exclusivity of events (intimacy). For example, in terms of the partial order, it is
not clear whether $u$ and $v$ having in common two events with two people
attending them is better than or worse than $u$ and $v$ having three events in
common with three people attending them.

We shall use the following definition for deciding how much total tie strength a
single event generates, given the size of the event.

\begin{notation} If there is a single event, with $k$ people, we shall denote
  the total tie-strength generated as $f(k)$.
\end{notation}

\begin{lemma}[Local Neighborhood] \label{lem:local_neighborhood} The \ts of $u$
  and $v$ is affected only by events that both $u$ and $v$ attend.
\end{lemma}
\begin{proof}
  Given a graph $G$ and users $u$ and $v$ in $G$, $G^{-u}$ is obtained by
  deleting all events that $u$ is not a part of. Similarly, $G^{-u,v}$ is
  obtained by deleting all events of $G^{-u}$ that $v$ is not a part of. By
  \ref{axiom:depend}, tie strength of $u$ only depends on events that $u$
  attends. Hence, $TS_G(u,v) = TS_{G^{-u}}(u,v)$. Also, tie strength of $v$ only
  depends on events that $v$ attends. Hence, $TS_G(u,v) = TS_{G^{-u}}(u,v) =
  TS_{G^{-u,v}}(u,v)$. This proves our claim.
\end{proof}

\begin{lemma} The tie strength between any two people is always non-negative and
  is equal to zero if they have never attended an event together.
\end{lemma}
\begin{proof} If two people have never attended an event together, then from
  Lemma~\ref{lem:local_neighborhood} the tie strength remains unchanged if we
  delete all the events not containing either which in this case is all the
  events. Then \ref{axiom:baseline} tells us that $TS(u,v) = 0$.

  Also, \ref{axiom:more_stronger} implies that $TS_G(u,v) \geq TS_\phi(u,v) =
  0$. Hence, the tie strength is always non-negative.
\end{proof}

\begin{lemma} If there is a single party, with $k$ people, the Tie Strength of
  each tie is equal to $\frac{f(k)}{{k\choose 2}}$.
\end{lemma}
\begin{proof} By \ref{axiom:iso}, it follows that the tie-strength on each tie is
  the same. Since the sum of all the ties is equal to $f(k)$, and there are ${k
    \choose 2}$ edges, the tie-strength of each edge is equal to $\frac{f(k)}{{k
      \choose 2}}$.
\end{proof}

\begin{lemma} 
  \label{lemma:ts_bound}
  The total tie strength created at an event $E$ with $k$ people is a monotone
  function $f(k)$ that is bounded by $1 \leq f(k) \leq {k \choose 2}$
\end{lemma}
\begin{proof}
  By \ref{axiom:smaller_stronger}, the tie strength of $u$ and $v$ due to $E$ is
  less than that if they were alone at the event. $TS_E(u,v) \leq TS_{u,v}(u,v)
  = 1$, by the Baseline axiom. Summing up over all ties gives us that
  $\sum_{u,v} TS_E(u,v) \leq {k \choose 2}$. Also, since larger events generate
  more ties, $f(k) \geq f(i) : \forall i < k$. In particular, $f(k) \geq f(1) =
  1$. This proves the result.
\end{proof}

We are now ready to state the main theorem in this section.
\begin{theorem} 
  \label{thm:char}
  Given a graph $G=(L \cup R, E)$ and two vertices $u,v$, if the tie-strength
  function $TS$ follows \axioms, then the function has to be of the form
  $$TS_G(u,v) = g(h(|P_1|), h(|P_2|), \ldots ,h(|P_k|) )$$ 
  where $\{P_i\}_{1\leq i\leq k}$ are the events common to both $u$ and $v$,
  $h:\natnum \to \real$ is a monotonically decreasing function bounded by $1
  \geq h(n) \geq \reci{{n \choose 2}}$ and $g:\natnum^* \to \real$ is a
  monotonically increasing submodular function.
\end{theorem}
\begin{proof}
  Given two users $u$ and $v$ we use \axioms\; to successively change the form
  of $TS_G(u,v)$. Let $\{P_i\}_{1\leq i\leq k}$ be all the events common to $u$
  and $v$. \ref{axiom:indep} implies that $TS_G(u,v) = g(TS_{P_i}(u,v))_{1\leq i
    \leq k}$, where $g$ is a monotonically increasing submodular function. Given
  an event $P$, $TS_P(u,v) = h(|P|) = \frac{f(|P|)}{ {|P| \choose 2}}$. By
  \ref{axiom:smaller_stronger}, $h$ is a monotonically decreasing function. Also,
  by Lemma~\ref{lemma:ts_bound}, $f$ is bounded by $1 \leq f(n) \leq {n \choose
    2}$. Hence, $h$ it bounded by $1 \geq h(n) \geq \reci{{n \choose 2}}$. This
  completes the proof of the theorem.
\end{proof}

Theorem~\ref{thm:char} gives us a way to explore the space of valid functions
for representing tie strength and find which work given particular
applications. In Section~\ref{sec:measures} we shall look at popular measure of
tie strength and show that most of them follow \axioms\; and hence are of the
form described by Theorem~\ref{thm:char}. We also describe the functions $h$ and
$g$ that characterize these common measures of \ts. While Theorem~\ref{thm:char}
gives a characterization of functions suitable for describing tie strength, they
leave open a wide variety of functions. In particular, it does not give the
comfort of having a single function that we could use. We discuss the reasons
for this and what we would need to do to settle upon a particular function in
the next section.

\subsection{Tie Strength and Orderings}
\label{sec:order}
We begin this section with a definition of order in a set.

\begin{definition}[Total Order] Given a set $S$ and a binary relation $\leq_\Oo$
  on $S$, $\Oo = (S, \leq_\Oo)$ is called a total order if and only if it
  satisfies the following properties
  \begin{inparaenum}[(i]
  \item Total). for every $u,v \in S$, $u \leq_\Oo v$ or $v\leq_\Oo u$
  \item Anti-Symmetric). $u \leq_\Oo v \text{ and } v \leq_\Oo u \implies
    u=v$
  \item Transitive). $u \leq_\Oo v$ and $v \leq_\Oo w \implies u\leq_\Oo w$
  \end{inparaenum}

  A total order is also called a linear order.
\end{definition}

Consider a measure $TS$ that assigns a measure of tie strength to each pair of
nodes $u,v$ given the events that all nodes attend in the form of a graph
$G$. Since $TS$ assigns a real number to each edge and the set of reals is
totally ordered, $TS$ gives a total order on all the edges. In fact, the
function $TS$ actually gives a total ordering of $\natnum^*$. In particular, if we fix
a vertex $u$, then $TS$ induces a total order on the set of neighbors of $u$,
given by the increasing values of $TS$ on the corresponding edges.

\subsubsection{The Partial Order on $\natnum^*$}
\label{sec:partial_order}

\begin{definition}[Partial Order] Given a set $S$ and a binary relation
  $\leq_\Po$ on $S$, $\Po = (S, \leq_\Po)$ is called a partial order if and only
  if it satisfies the following properties
  \begin{inparaenum}[(i]
  \item Reflexive). for every $u \in S, u \leq_\Po u$
  \item Anti-Symmetric). $u \leq_\Po v \text{ and } v \leq_\Po u \implies
    u=v$
  \item Transitive). $u \leq_\Po v$ and $v \leq_\Po w \implies u\leq_\Po w$
  \end{inparaenum}

  The set $S$ is called a partially ordered set or a poset.
\end{definition}

Note the difference from a total order is that in a partial order not every pair
of elements is comparable. We shall now look at a natural partial order $\N =
(\natnum^*, \leq_\N)$ on the set $\natnum^*$ of all finite sequences of natural
numbers. Recall that $\natnum^* = \cup_k \natnum^k$. We shall think of this
sequence as the number of common events that a pair of users attend.

\begin{definition}[Partial order on $N^*$] 
  Let $a,b \in \natnum^*$ where $a = (a_i)_{1\leq i \leq A}$ and $b =
  (b_i)_{1\leq i \leq B}$. We say that $a \geq_\N b$ if and only if $A \geq B$
  and $a_i \leq b_i: 1 \leq i \leq B$. This gives the partial order $\N =
  (\natnum^*, \leq_\N)$.
\end{definition}

The partial order $\N$ corresponds to the intuition that more events and smaller
events create stronger ties. In fact, we claim that this is exactly the partial
order implied by the \axioms. Theorem~\ref{thm:ts_partial_order} formalizes this
intuition along with giving the proof. What we would really like is a total
ordering. Can we go from the partial ordering given by the \axioms\; to a total
order on $\natnum^*$? Theorem~\ref{thm:ts_partial_order} also suggest ways in
which we can do this.

\subsubsection{Partial Orderings and Linear Extensions}
\label{sec:linear_ext}

In this section, we connect the definitions of partial order and the functions
of tie strength that we are studying. First we start with a definition.

\begin{definition}[Linear Extension] $\Lo = (S, \leq_\Lo)$ is called the
  \emph{linear extension} of a given partial order $\Po = (S, \leq_\Po)$ if and
  only if $\Lo$ is a total order and $\Lo$ is consistent with the ordering
  defined by $\Po$, that is, for all $u,v \in S$, $u \leq_\Po v \implies u
  \leq_\Lo v$.
\end{definition}

We are now ready to state the main theorem which characterizes functions that
satisfy \axioms\; in terms of a partial ordering on $\natnum^*$. Fix nodes $u$
and $v$ and let $P_1,\ldots,P_n$ be all the events that both $u$ and $v$
attend. Consider the sequence of numbers $(|P_i|)_{1 \leq i \leq k}$ that give
the number of people in each of these events. Without loss of generality assume
that these are sorted in ascending order. Hence $|P_i| \leq |P_{i+1}|$. We
associate this \emph{sorted sequence} of numbers with the tie $(u,v)$. The
partial order $\N$ induces a partial order on the set of pairs via this
mapping. We also call this partial order $\N$. Fixing any particular measure of
tie strength, gives a mapping of $\natnum^*$ to $\real$ and hence implies
fixing a particular linear extension of $\N$, and fixing a linear extension of
$\N$ involves making non-obvious decisions between elements of the partial
order. We formalize this in the next theorem.

\begin{theorem}
  \label{thm:ts_partial_order}
  Let $G = (L \cup R, E)$ be a bipartite graph of users and events. Given two
  users $(u,v) \in (L \times L)$, let $(|P_i|)_{1 \leq i \leq k} \in R$ be the
  set of events common to users $(u,v)$. Through this association, the partial
  order $\N = (\natnum^*, \leq_\N)$ on finite sequences of numbers induces a
  partial order on $L \times L$ which we also call $\N$.

  Let $TS$ be a function that satisfies \axioms. Then $TS$ induces a total order
  on the edges that is a linear extension of the partial order $\N$ on $L \times
  L$.

  Conversely, for every linear extension $\Lo$ of the partial order $\N$, we can
  find a function $TS$ that induces $\Lo$ on $L\times L$ and that satisfies \axioms.

\end{theorem}
\begin{proof}
  $TS: L \times L \to \real$. Hence, it gives a total order on the set of pairs
  of user. We want to show that if $TS$ satisfies \axioms, then the total order
  is a linear extension of $\N$. The characterization in Theorem~\ref{thm:char}
  states that given a pair of vertices $(u,v) \in (L \times L)$, $TS(u,v)$ is
  characterized by the number of users in events common to $u$ and $v$ and can
  be expressed as $TS_G(u,v) = g(h(|P_i|))_{1\leq i \leq k}$ where $g$ is a
  monotone submodular function and $h$ is a monotone decreasing function. Since
  $TS : L\times L \to \real$, it induces a total order on all pairs of users. We
  now show that this is a consistent with the partial order $\N$. Consider two
  pairs $(u_1,v_1) , (u_2, v_2)$ with party profiles $a = (a_1,\ldots,a_A)$ and
  $b = (b_1, \ldots, b_B)$.

  Suppose $a \geq_\N b$. We want to show that $TS(u_1,v_1) \geq TS(u_2,v_2)$. $a
  \geq_\N b$ implies that $A \geq B$ and that $a_i \leq b_i : \forall 1 \leq i
  \leq B$.
  \begin{align*}
    &TS(u_1,v_1)  \\
    &= g(h(a_1), \ldots, h(a_A)) \\
    &\geq g(h(a_1), \ldots, h(a_B)) \text{ (Since $g$ is monotone and $A \geq B$) }\\
    &\geq g(h(b_1), \ldots, h(b_B)) \text{ (Since $g$ is monotone and } \\
    & \qquad h(a_i) \geq h(b_i) \text{ since } a_i \leq b_i ) \\
    &= TS(u_2,v_2)
  \end{align*}
  This proves the first part of the theorem.

  For the converse, we are given an total ordering $\Lo = (\natnum^*,\leq_\Lo)$
  that is an extension of the partial order $\N$. We want to prove that there
  exists a tie strength function $TS: L \times L \to \real$ that satisfies
  Axioms (1-6) and that induces $\Lo$ on $L \times L$. We shall prove this by
  constructing such a function. We shall define a function $f:\natnum^* \to
  \rational$ and define $TS_G(u,v) = f(a_1,\ldots,a_k)$, where $a_i = |P_i|$,
  the number of users that attend event $P_i$ in $G$.

  Define $f(n) = \reci{n-1}$ and $f(\phi) = 0$. Hence, $TS_\phi(u,v) = f(\phi) =
  0$ and $TS_{\{u,v\}}(u,v) = f(2) = \reci{2-1} = 1$. This shows that $TS$
  satisfies \ref{axiom:baseline}. Also, define $f(\underbrace{1, 1, \ldots,
    1}_{n}) = n$.  Since $\natnum^*$ is countable, consider elements in some
  order. If for the current element $a$ under consideration, there exists an
  element $b$ such that $a =_\N b$ and we have already defined $TS(b)$, then
  define $TS(a) = TS(b)$. Else, find let $a_{glb} = argmax_e\left\{TS(e) \text{
      is defined and } a \geq_\N e\right\}$ and let $a_{lub} = \argmin_e \left\{
    TS(e) \text{ is defined and } a \leq_\N e\right\}$. Since, at every point
  the sets over which we take the maximum of minimum are finite, both $a_{glb}$
  and $a_{lub}$ are well defined and exist. Define $TS(a) = \half \left(
    TS(a_{glb}) + TS(a_{lub})\right)$.
\end{proof}

In this abstract framework, an intuitively appealing linear extension is the
random linear extension of the partial order under consideration. There are
polynomial time algorithms to calculate
this~\citep*{random_linear_extension}. We leave the analysis of the analytical
properties and its viability as a strength function in real world applications
as an open research question.

In the next section, we turn our attention to actual measures of tie
strength. We see some popular measures that have been proposed before as well as
some new ones.

\section{Measures of tie strength}
\label{sec:measures}

\begin{table*}[ht]
  \centering
  \begin{tabularx}{\textwidth}{|l|c|c|c|c|c|c|c|c|X|}
    \hline
    \backslashbox{Measures of Tie Strength \qquad}{Axioms} &
    \begin{turn}{90} \ref{axiom:iso} \end{turn} &
    \begin{turn}{90} \ref{axiom:baseline} \end{turn} &
    \begin{turn}{90} \ref{axiom:more_stronger} \end{turn} &
    \begin{turn}{90} \ref{axiom:smaller_stronger} \end{turn} &
    \begin{turn}{90} \ref{axiom:larger_more} \end{turn} &
    \begin{turn}{90} \ref{axiom:depend} \end{turn} &
    \begin{turn}{90} \ref{axiom:indep} \end{turn} &
    \begin{turn}{90} \ref{axiom:submodularity} \end{turn} &
    $g(a_1,\ldots,a_k)$ and $h(|P_i|)=a_i$ \newline(From the characterization in \newline Theorem~\ref{thm:char})\\
    \hline \ref{measure:common} &\checkmark &\checkmark &\checkmark &\checkmark &\checkmark &\checkmark &\checkmark &\checkmark 
    & $g(a_1,\ldots,a_k) = \sum_{i=1}^ka_i \newline h(n) = 1$\\
    \hline
    \ref{measure:jaccard} &\checkmark &\checkmark &\checkmark &\checkmark &\checkmark &x &x &x &x \\
    \hline 
    \ref{measure:delta} &\checkmark &\checkmark &\checkmark &\checkmark &\checkmark &\checkmark &\checkmark &\checkmark 
    & $g(a_1,\ldots,a_k) = \sum_{i=1}^ka_i \newline h(n) = \reci{{n\choose 2}}$\\
    \hline 
    \ref{measure:adamic_adar} &\checkmark &\checkmark &\checkmark &\checkmark &\checkmark &\checkmark &\checkmark &\checkmark &
    $g(a_1,\ldots,a_k) = \sum_{i=1}^ka_i \newline h(n) = \reci{\log n}$\\
    \hline
    \ref{measure:katz} &\checkmark &x &\checkmark &\checkmark &\checkmark &\checkmark &x &x &x\\
    \hline 
    \ref{measure:pref} &\checkmark &\checkmark &x &\checkmark &\checkmark &\checkmark &x &x &x\\
    \hline
    \ref{measure:random_walk} &\checkmark &x &x &x &\checkmark &\checkmark &x &x &x\\
    \hline
    \ref{measure:simrank} &\checkmark &x &x &x &x &x &x &x &x\\
    \hline 
    \ref{measure:max} &\checkmark &\checkmark &\checkmark &\checkmark &\checkmark &\checkmark &\checkmark &\checkmark 
    & $g(a_1,\ldots,a_k) = \max_{i=1}^ka_i \newline h(n) = \reci{n}$\\
    \hline 
    \ref{measure:newman} &\checkmark &\checkmark &\checkmark &\checkmark &\checkmark &\checkmark &\checkmark &\checkmark 
    & $g(a_1,\ldots,a_k) = \sum_{i=1}^ka_i \newline h(n) = \reci{n}$ \\
    \hline
    \ref{measure:prop} &\checkmark &x &x &\checkmark &x &\checkmark &x &x & x\\
    \hline
  \end{tabularx}
  \caption{Measures of \ts and the axioms they satisfy}
  \label{tab:measures}
\end{table*}

There have been plenty of tie-strength measures discussed in previous literature. We review the most popular of them here and classify them according to the axioms they satisfy. In this section, for an event $P$, we denote by $|P|$ the number of people in the event $P$.  The size of $P$'s neighborhood is represented by $|\Gamma(P)|$.

\begin{description}
\item[Common Neighbors. \label{measure:common}] This is the simplest measure of
  tie strength, given by the total number of common events that both $u$ and $v$
  attended. $$TS(u, v) = |\Gamma(u)\cap \Gamma(v)| $$
  
\item[Jaccard Index. \label{measure:jaccard}] A more refined measure of tie
  strength is given by the Jaccard Index, which normalizes for how ``social'' $u$ and $v$ are
  $$ TS(u, v) = \frac{|\Gamma(u)\cap \Gamma(v)|} {|\Gamma(u)\cup \Gamma(v)|} $$
  
\item[Delta.\label{measure:delta}] Tie strength increases with the number of events.  $$TS(u, v) = \sum_{P \in \Gamma(u)\cap
    \Gamma(v)} \reci{{|P| \choose 2}}$$
    
\item[Adamic and Adar.\label{measure:adamic_adar}] This measure was introduced
  in~\citep*{adamic_score}.
 $$ TS(u, v) = \sum_{P \in \Gamma(u)\cap \Gamma(v)} \reci{\log|P|}$$
 
\item[Linear.\label{measure:newman}] Tie strength increases with number of events.
  $$TS(u, v) = \sum_{P \in \Gamma(u)\cap \Gamma(v)} \reci{|P|}$$ 
   
\item[Preferential attachment.\label{measure:pref}] $$ TS(u,v) =
  |\Gamma(u)|\cdot |\Gamma(v)|$$
  
\item[Katz Measure.\label{measure:katz}] This was introduced in \citep*{katz}. It
  counts the number of paths between $u$ and $v$, where each path is discounted
  exponentially by the length of path.
  $$ TS(u,v) = \sum_{q \in \text{ path between } u,v} \gamma^{-|q|}$$
  
\item[Random Walk with Restarts.\label{measure:random_walk}] This gives a
  non-symmetric measure of tie strength. For a node $u$, we jump with probability
  $\alpha$ to node $u$ and with probability $1-\alpha$ to a neighbor of the
  current node. $\alpha$ is the restart probability.  The \ts between $u$ and $v$ is the stationary probability that
  we end at node $v$ under this process.
  
\item[Simrank.\label{measure:simrank}] This captures the similarity between two
  nodes $u$ and $v$ by recursively computing the similarity of their neighbors.
  \[ TS(u,v) = \begin{cases}
    1 & \text{ if  } u = v \\
    \gamma \cdot \frac{ \sum_{a \in \Gamma(u)}\sum_{b \in \Gamma(v)} TS(a,b)}
    {|\Gamma(u)|
      \cdot |\Gamma(v)|} & \text{ otherwise}\\
    \end{cases}
    \]
\end{description}

Now, we shall introduce three new measures of tie strength. 
 In a sense, $g = \sum$ is at one extreme
of the range of functions allowed by Theorem~\ref{thm:char} and that is the
default function used. $g = \max$ is at the other extreme of the range of
functions. 

\begin{description}
\item[Max.\label{measure:max}] Tie strength does not increases with number of events
  $$TS(u, v) = \max_{P \in \Gamma(u)\cap \Gamma(v)} \reci{|P|} $$. 
  
\item[Proportional.\label{measure:prop}] Tie strength increases with number of events. People spend time
  proportional to their TS in a party. 
  TS is the fixed point of this set of equations:
  
  $$TS(u,v)=\sum_{P \in \Gamma(u)\cap \Gamma(v)} { \frac{\epsilon}{|P|} +
    (1-\epsilon) \frac{TS(u,v)}{ \sum_{w \in \Gamma(u)} TS(u,w)} }$$
  
\item[Temporal Proportional.\label{measure:temporal_prop}] This is similar to Proportional,
  but with a temporal aspect.  $TS$ is not a fixed point, but starts with a
  default value and is changed according to the following equation, where the events are
  ordered by time.
  \begin{align*} &TS(u,v,t) \\
  &= \begin{cases}
    TS(u,v,t-1) \quad \text{ if $u$ and $v$ do not attend }P_t\\
    \epsilon\frac{1}{|P_t|} + (1-\epsilon) \frac{TS(u,v,t-1)} {\sum_{w \in P_t}
     TS(u,w,t-1)} \quad \text{ otherwise }
  \end{cases}\\
  \end{align*}
\end{description}

Table~\ref{tab:measures} provides a classification of all these tie-strength measures, according to which axioms they satisfy. If they satisfy all the axioms, then we use Theorem~\ref{thm:char} to find the characterizing functions $g$ and $h$.  An interesting observation is that all the ``self-referential'' measures (such as Katz Measure, Random Walk with Restart, Simrank, and Proportional) fail to satisfy the axioms.  Another interesting observation is in the classification of measures that satisfy the axioms.  The majority use $g=\sum$ to aggregate tie strength across events.  Per event, the majority compute tie strength as one over a simple function of the size of the party.

\begin{figure*}
  \centering
  \begin{tabular}{ccc}
    \includegraphics[scale=0.26]{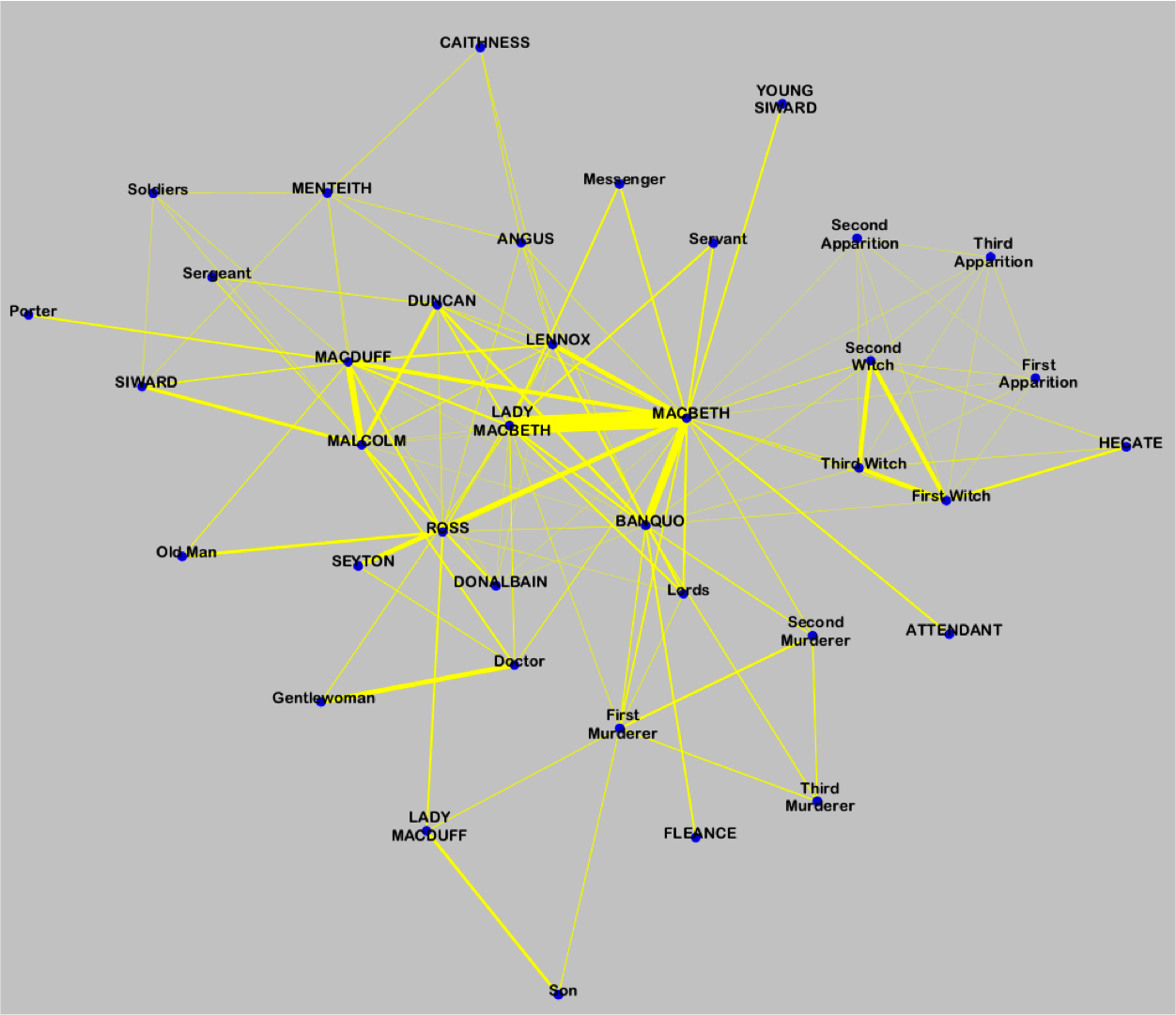} & 
    \includegraphics[scale=0.26]{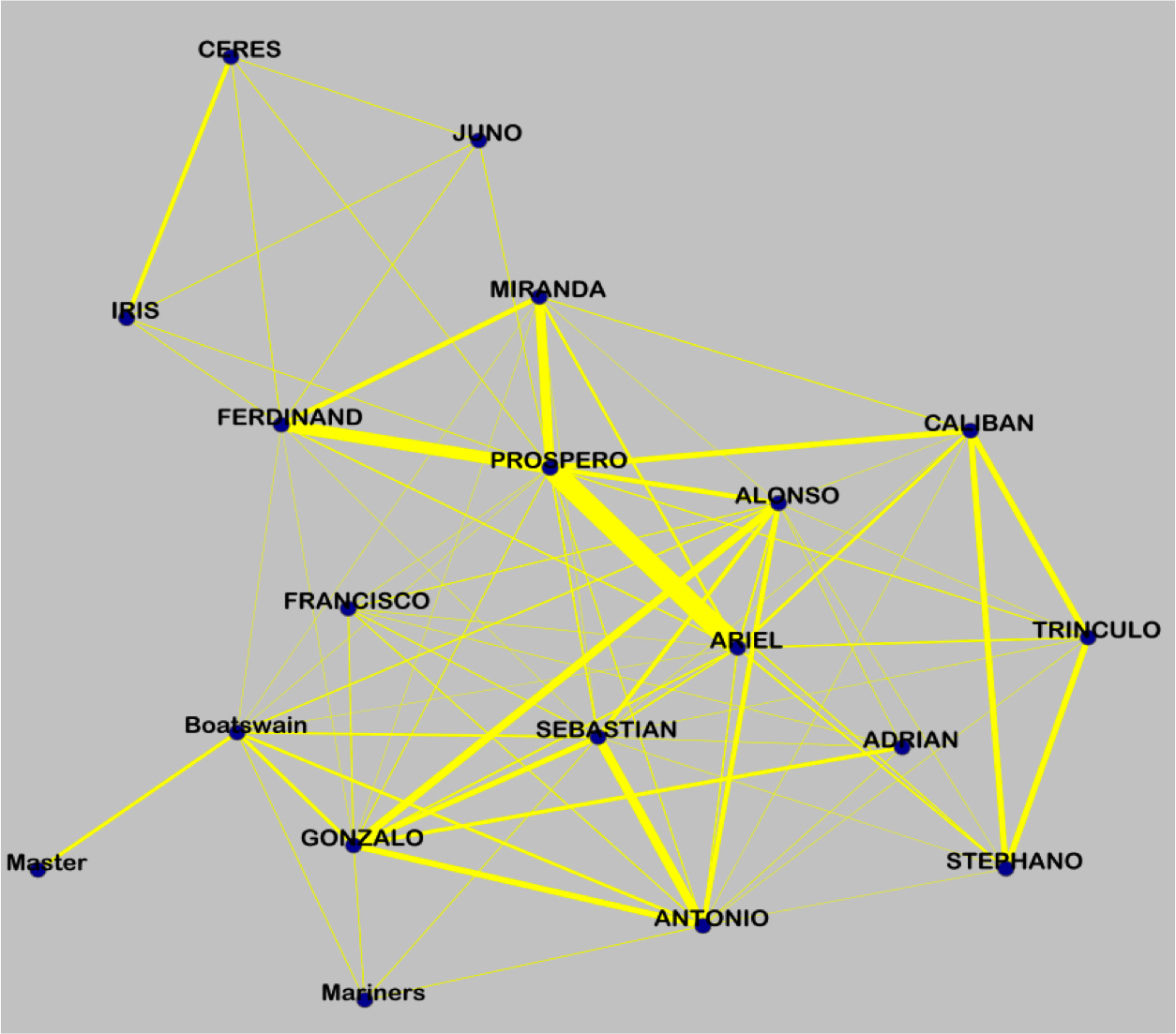} & 
    \includegraphics[scale=0.26]{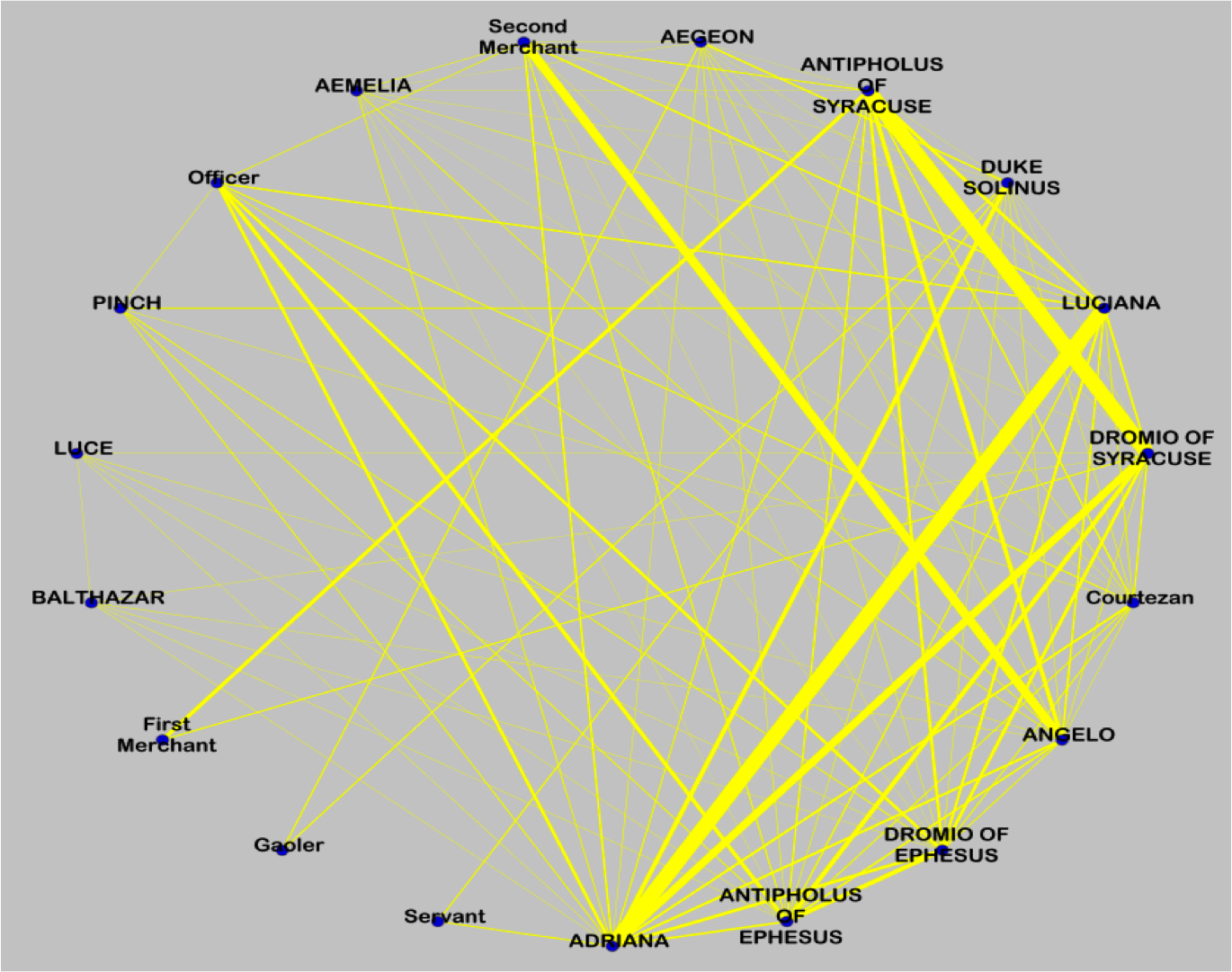} \\ 
    {\bf Macbeth} &  {\bf Tempest} & {\bf A Comedy of Errors}  \\
  \end{tabular}
  \caption{Inferred weighted social networks between characters in Shakespearean plays.  The thicker an
    edge, the stronger the tie. Tie Strength was calculated using the tie-strength measure~\ref{measure:newman}}
  \label{fig:macbeth}
\end{figure*}

\section{Experiments}
\label{sec:experiments}
This section presents our findings on five data sets: Shakespearean plays
(Macbeth, Tempest, and A Comedy of Errors), Reality Mining, and Enron Emails.

\subsection{Data Sets}
\paragraph{Shakespearean Plays}
We take three well-known plays by Shakespeare (Macbeth, Tempest, and A Comedy of
Errors) and create bipartite person$\times$event graphs. The person-nodes are the
characters in the play. Each event is a set of characters who are on the stage
at the same time. We calculate the strength of ties between each pair of
nodes. Thus without using any semantic information and even without analyzing
any dialogue, we estimate how much characters interact with one another.

\paragraph{The Reality Mining Project}
This is the popular dataset from the Reality Mining project at
MIT~\citep*{reality_mining}. This study gave one hundred smart phones to
participants and logged information generated by these smart phones for several
months. We use the bluetooth proximity data generated as part of
this project. The bluetooth radio was switched on every five minutes and logged
other bluetooth devices in close proximity. The people are the participants in
the study and events record the proximity between people. This gives us a
total of 326,248 events.  

\paragraph{Enron Emails}
This dataset consists of emails from 150 users from the Enron corporation, that were
made public during the Federal Energy Regulatory Commission investigation. We
look at all emails that occur between Enron addresses. Each email is an event
and all the people copied on that email i.e. the sender (from), the receivers (to, cc
and bcc) are included in that event. This gives a total of 32,471 people and
371,321 events. 

\subsection{Measuring Coverage of the Axioms}
In Section~\ref{sec:axioms}, we discussed axioms governing tie-strength and
characterized the axioms in terms of a partial order in
Theorem~\ref{thm:ts_partial_order}. We shall now look at an experiment to
determine the ``coverage'' of the axioms, in terms of the number of pairs of
ties that are actually ordered by the partial order. 

For different datasets, we use Theorem~\ref{thm:ts_partial_order} to generate a
partial order between all ties. Table~\ref{table:partial_order} shows the
percentage of all ties that are \emph{not} resolved by the partial order --  i.e., \emph{the partial
order cannot tells us if one tie is greater or if they are equal}. Each measure of
tie-strength gives a total order on the ties; and, hence resolves all the
comparisons between pairs of ties. The number of tie-pairs which are left
incomparable in the partial order gives a notion of the how much room the axioms
leave open for different tie-strength functions to differ from each
other. Table~\ref{table:partial_order} shows that partial order \emph{does}
resolve a very high percentage of the ties. Also, we see that real-world
datasets (e.g., Reality Mining) have more unresolved ties than the cleaner Shakespearean plays
datasets.

\begin{table}[ht]
\begin{tabular}{l c r}
\hline
Dataset & Tie Pairs & Incomparable Pairs (\%) \\
\hline
Tempest & 14,535 & 275 (1.89) \\
Comedy of Errors & 14,535 & 726  (4.99) \\
Macbeth & 246,753 & 584 (0.23) \\
Reality Mining & 13,794,378 & 1,764,546 (12.79) \\
\hline
\end{tabular}
\caption{Number of ties \emph{not} resolved by the partial order.  The last column shows the percentage of tie pairs on which different tie-strength functions can differ.}
\label{table:partial_order}
\end{table}

Next, we look at two tie-strength functions (Jaccard Index and
Temporal Proportional) which do not obey the
axioms. As previously shown, Theorem~\ref{thm:ts_partial_order}  implies that these functions do
not obey the partial order.  So, there are some tie-pairs in conflict with the
partial order. Table~\ref{table:partial_order_conflict} shows the number of tie-pairs that are actually in conflict. This experiment 
gives us some intuition about how far away a measure is from the axioms. We see
that for these datasets, Temporal Proportional agrees with
the partial order more than the Jaccard Index.  
 We can also see that as the size of the dataset increases, the percentage of
conflicts decreases drastically.

\begin{table}[ht]
\centering
\begin{tabular}{l c r r}
\hline
Dataset & Tie Pairs & Jaccard (\%) & Temporal(\%) \\
\hline
Tempest & 14,535 & 488 (3.35) & 261 (1.79)\\
Comedy & 14,535  & 1,114 (7.76)& 381 (2.62)\\
Macbeth & 246,753 & 2,638 (1.06)& 978 (0.39) \\
Reality & 13,794,378 & 290,934 (0.02) & 112,546 (0.01) \\
\hline
\end{tabular}
\caption{Number of conflicts between the partial order and the tie-strength functions: Jaccard Index and Temporal Proportional.  The second and third columns show the percentage of tie-pairs in conflict with the partial order.}
\label{table:partial_order_conflict}
\end{table}

\subsection{Visualizing Networks}

We obtain the tie strength between characters from Shakespearean plays using the
linear function proposed by~\ref{measure:newman} Figure~\ref{fig:macbeth} shows
the inferred weighted social networks. Note that the inference is only based on people
occupying the same stage and not on any semantic analysis of the text. The
inferred weights (i.e. tie strengths) are consistent with the stories. For example, the highest
tie strengths are between Macbeth and Lady Macbeth in the play Macbeth, between
Ariel and Prospero in Tempest, and between Dromio of Syracuse and Antipholus of
Syracuse in A Comedy of Errors.

\subsection{Measuring Correlation among Tie-Strength Functions}

Figures~\ref{fig:degree_distribution} and~\ref{fig:shakes_degree_distribution} show the frequency distributions of the number of people at an event. We see that these distributions are very different for the different graphs (even among the real-world communication networks, Enron and MIT Reality Mining). This suggests that different applications
might need different measures of tie strength.

\begin{figure}
  \centering
  \includegraphics[width=0.43\textwidth]{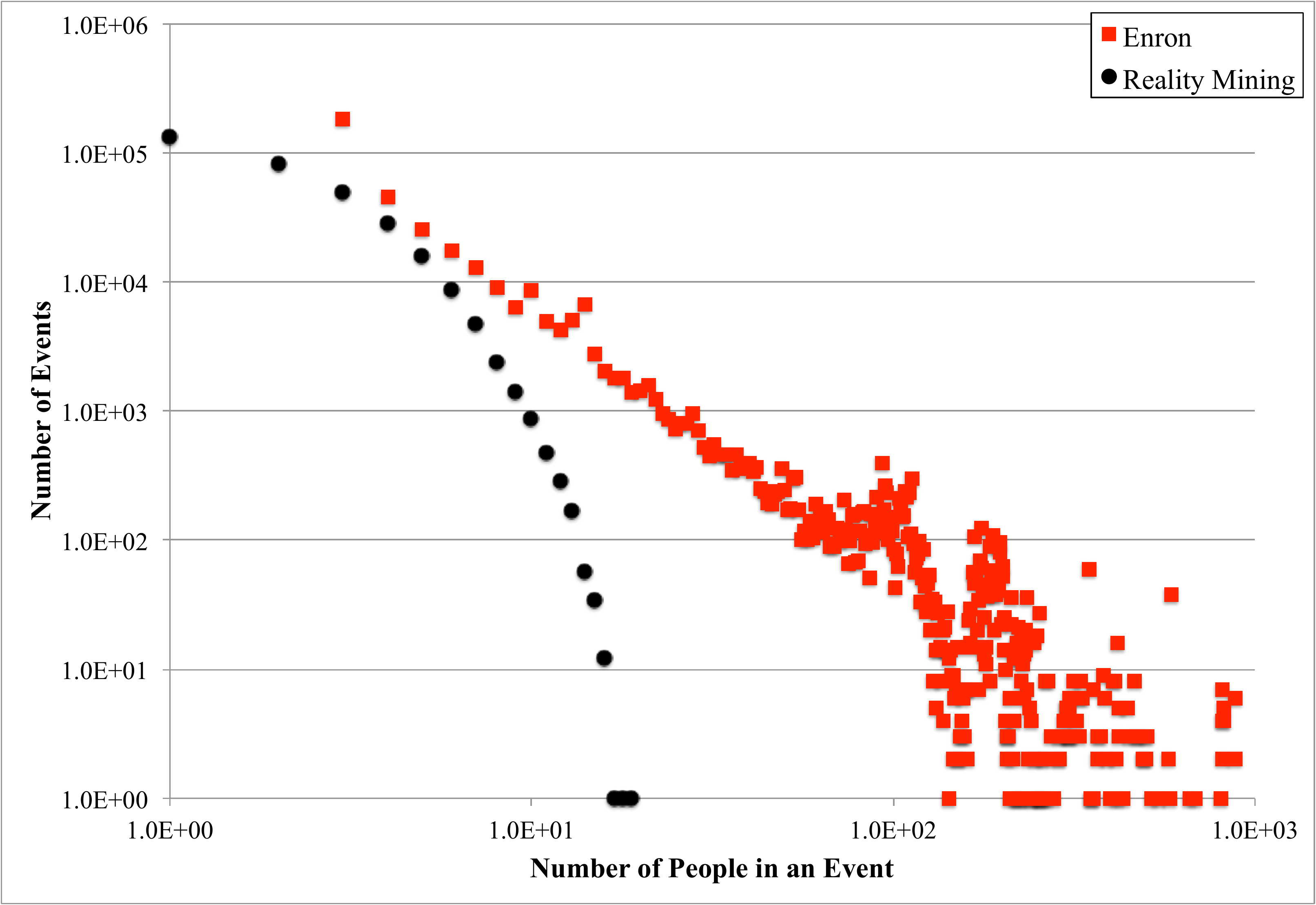}
  \caption{Frequency distribution of number of people per event for the Reality
    Mining and Enron datasets (in log-log scale)}
  \label{fig:degree_distribution}
\end{figure}

\begin{figure}
  \centering
  \includegraphics[width=0.43\textwidth]{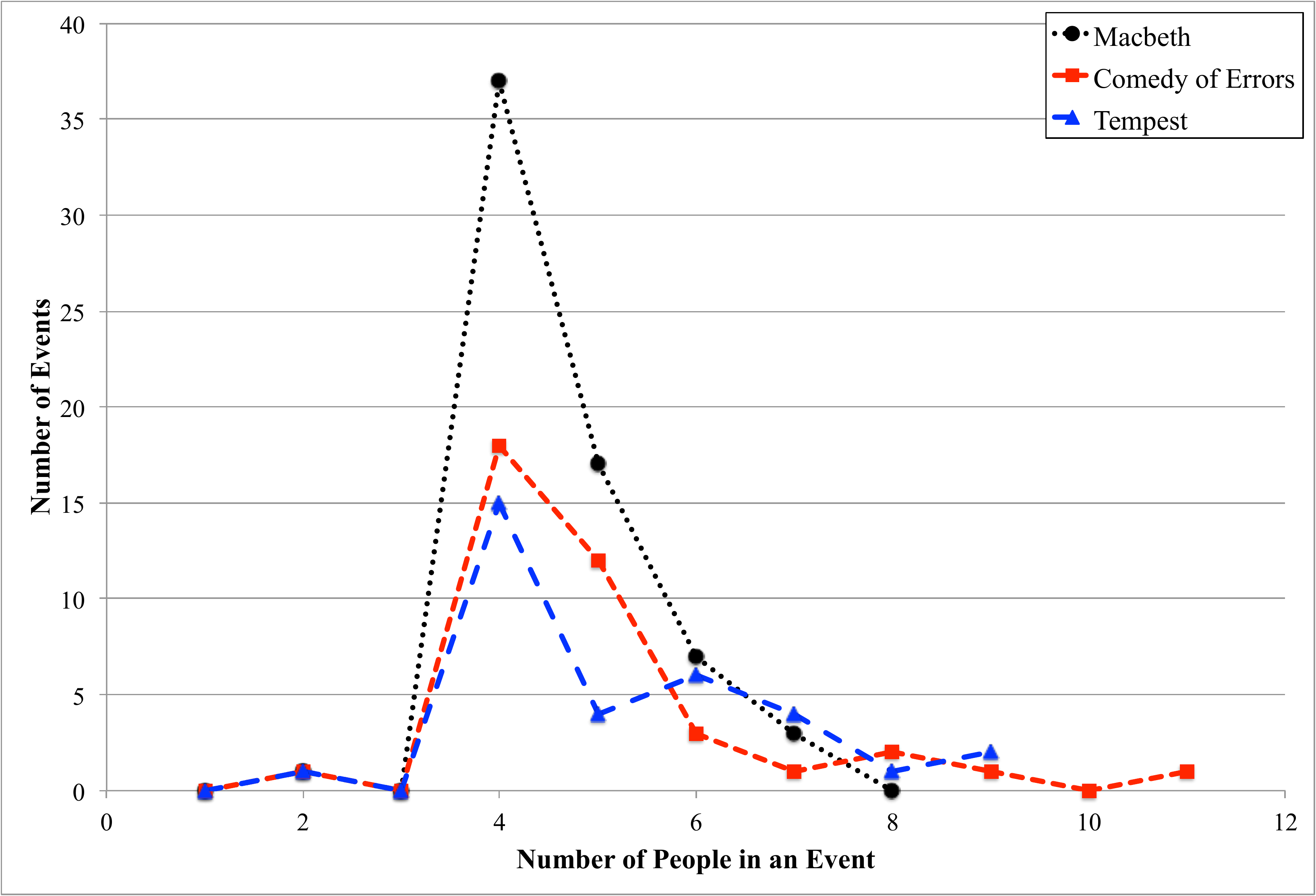}
  \caption{Frequency distribution of number of people per event for the Shakespearean Plays}
  \label{fig:shakes_degree_distribution}
\end{figure}

Figure 4 shows Kendall's $\tau$ coefficient for the Shakespearean plays, the
Reality Mining data and Enron emails.  Depending on the data set, different
measures of tie strength are correlated.  For instance, in the ``clean'' world
of Shakespearean plays Common Neighbor is the least correlated measure; while in the
``messy'' real world data from Reality Mining and Enron emails, Max is the least
correlated measure. 

\begin{figure*}[tbh]
  \centering
  \includegraphics[scale=1]{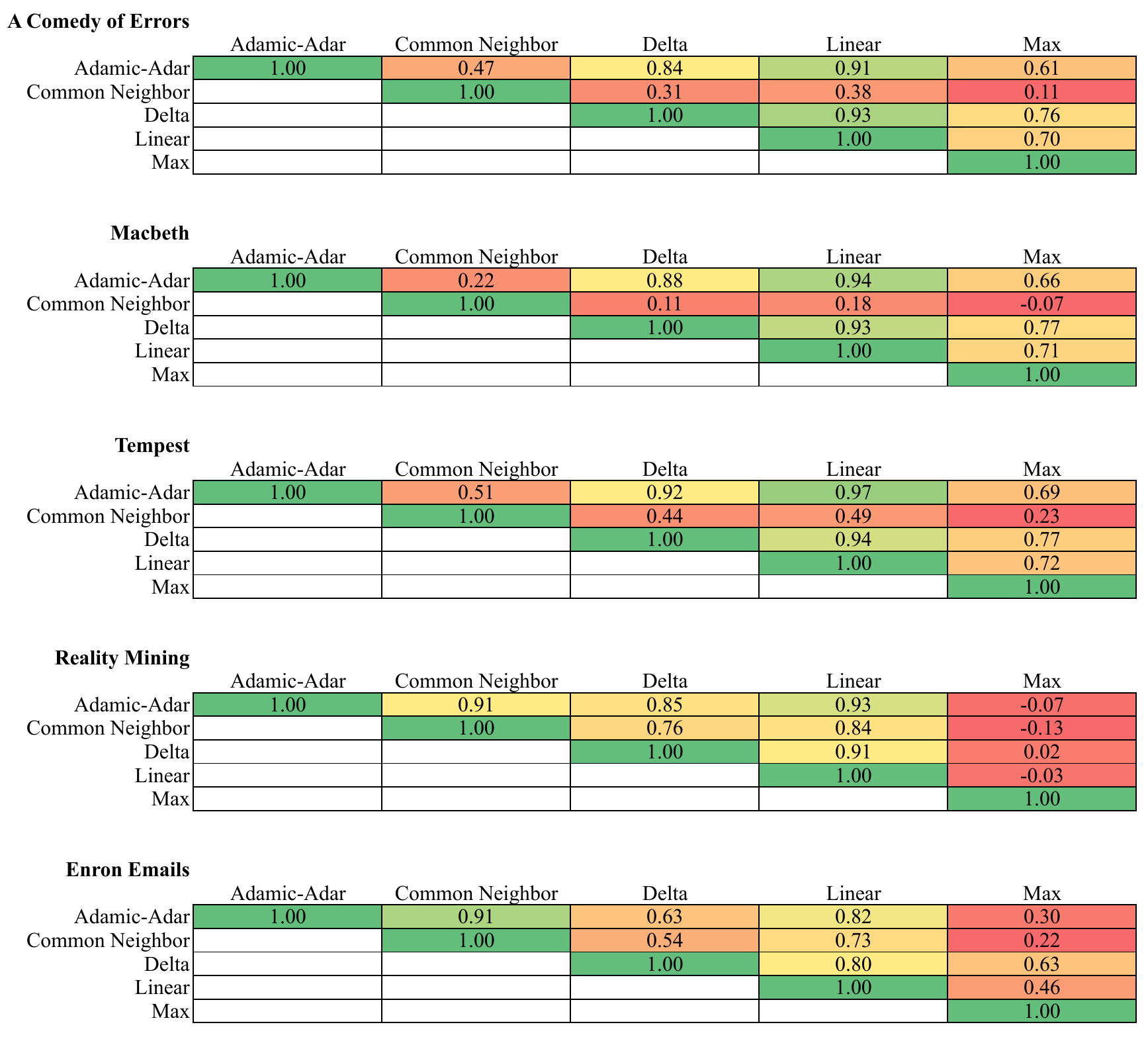}
  \label{fig:kendall_tau}
  \caption{Kendall's $\tau$ coefficient for Shakespearean plays, the Reality
    Mining data and the Enron emails.  The color scale goes from bright green (coefficient = 1) to bright red (coefficient = -1).  In the Skakespearean plays, the least correlated measure is Common Neighbor (as indicated by the red cells in that column).  In the real-world communication networks of Enron and Reality Mining, the least correlated measure is Max (again as indicated by the red cells in that column).  Since the correlation matrices are symmetric, we show only the upper-triangle entries.}
\end{figure*}

\section{Conclusions}
\label{sec:conclusions}

We presented an axiomatic approach to the problem of inferring implicit social
networks by measuring tie strength from bipartite person$\times$event graphs.  We
characterized functions that satisfy all axioms and demonstrated a range of measures
that satisfy this characterization.  We showed that in ranking applications, the
axioms are equivalent to a natural partial order; and demonstrated that to
settle on a particular measure, we must make a non-obvious decision about
extending this partial order to a total order which is best left to the
particular application.  We classified measures found in prior literature
according to the axioms that they satisfy.  Finally, our experiments demonstrated the coverage of the axioms and revealed through
the use of Kendall's Tau correlation whether a dataset is well-behaved, where we
do not have to worry about which tie-strength measure to choose, or we have to
be careful about the exact choice of measure.

\bibliographystyle{abbrvnat}
\bibliography{connected}  

\end{document}